\documentclass[showpacs,twocolumn,pra,superscriptaddress,notitlepage]{revtex4-1}
\usepackage{qcircuit}
\usepackage[dvips]{graphicx}
\usepackage{amsmath,amssymb,amsthm,mathrsfs,amsfonts,dsfont}
\usepackage{subfigure, epsfig}
\usepackage{braket}
\usepackage{bm}
\usepackage{enumerate}
\usepackage{color}
\usepackage{comment}
\usepackage[colorlinks = true]{hyperref}

\newtheorem{theorem}{Theorem}
\newtheorem{observation}{Observation}
\newtheorem{lemma}{Lemma}

\newcommand{\mc}{\mathcal}
\newcommand{\mb}{\mathbf}
\newcommand{\tr}{\mathrm{Tr}}

\newcommand{\red}[1]{{\color{red} #1}}
\newcommand{\blue}[1]{{\color{blue} #1}}

\begin{document}
\title{Quantum gate verification and its application in property testing}

\author{Pei Zeng}
\affiliation{Center for Quantum Information, Institute for Interdisciplinary Information Sciences, Tsinghua University, Beijing 100084, China}
\author{You Zhou}
\email{zyqphy@gmail.com}
\affiliation{Department of Physics, Harvard University, Cambridge, Massachusetts 02138, USA}
\affiliation{CAS Centre for Excellence and Synergetic Innovation Centre in Quantum Information and Quantum Physics, University of Science and Technology of China, Hefei, Anhui 230026, China}
\affiliation{Center for Quantum Information, Institute for Interdisciplinary Information Sciences, Tsinghua University, Beijing 100084, China}
\author{Zhenhuan Liu}
\affiliation{School of Physics, Peking University, Beijing 100871, China}
\affiliation{Center for Quantum Information, Institute for Interdisciplinary Information Sciences, Tsinghua University, Beijing 100084, China}
\begin{abstract}
To guarantee the normal functioning of quantum devices in different scenarios, appropriate benchmarking tool kits are quite significant. Inspired by the recent progress on quantum state verification, here we establish a general framework of verifying a target unitary gate. In both the non-adversarial and adversarial scenarios, we provide efficient methods to evaluate the performance of verification strategies for any qudit unitary gate. Furthermore, we figure out the optimal strategy and its realization with local operations. Specifically, for the commonly-used quantum gates like single qubit and qudit gates, multi-qubit Clifford gates, and multi-qubit generalized Controlled-Z(X) gates, we provide efficient local verification protocols. Besides, we discuss the application of gate verification to the detection of entanglement-preserving property of quantum channels and further quantify the robustness measure of them. We believe that the gate verification is a promising way to benchmark a large-scale quantum circuit as well as to test its property.

\end{abstract}
\maketitle

To build a large-scale and stable quantum system, efficient and robust benchmarking tools are essential \cite{eisert2019quantum}. The core aim of the quantum benchmarking is to establish the correct functioning of a quantum device, so that one can gain the confidence on the final information processing results. A benchmarking process is usually composed of several elements: the unknown target devices, some trusted (or partially characterized) benchmarking devices, and a benchmarking protocol with classical data processing.

While quantum mechanics endows us a large Hilbert space for information processing, whose size increases exponentially with the increase of the qubit number, it also introduces a challenging problem of characterizing the devices in this space. In general, without any prior knowledge on the target device, it on the same time takes exponentially increasing resources to get the full tomographic image of it \cite{Paris2004esimation,Vogel1989Determination}. Fortunately, in most of the cases, one holds some prior knowledge on the possible structure of the target device. With the assistance of this prior knowledge, it is in principle feasible to reduce the benchmarking resources and even characterize the system efficiently with a polynomial number of trials. Some common benchmarking tool kits developed in this spirit and widely applied in experiments are quantum tomography based on compressed sensing \cite{gross2010quantum,Flammia2012compressed}, tensor-network-based quantum tomography \cite{cramer2010efficient,Baumgratz2013Scalable,Lanyon2017Efficient}, permutation-invariant quantum tomography \cite{Toth2010Permutation,Tobias2012Permutation,Zhou2019Decomposition} and direct fidelity estimation \cite{flammia2011direct}, ordered by less information gain or higher efficiency.


On the other hand, the correctness of the benchmarking results usually relies on some assumptions made on the benchmarking devices as well as the target devices. In practice, the quantum gate benchmarking protocols with less assumptions on the benchmarking devices have been proposed, such as gate-set tomography \cite{merkey2013self,blumekohout2013robust} and randomized benchmarking \cite{emerson2005scalable,dankert2009exant,magesan2011scalable}, which can in some sense eliminate the effect of the state preparation and measurement error. Meanwhile, in some quantum information tasks such as quantum key distribution \cite{bennett1984quantum,lo1999unconditional} and blind quantum computation \cite{broadbent2009universal}, the quantum objects might be produced by some adversarial party, which may be correlated among different trials. Thus in these tasks one should make possibly less or no assumption on the target devices.  Currently, the protocol with the least device assumption both on benchmarking and target devices is the self-testing one \cite{brunner2014bell,Ivan2019self}, but is not efficient to extend to multi-partite system in general. As a result, robust benchmarking protocol against correlated noise is significant to explore for practical applications.

Recently, a highly efficient benchmarking protocol called quantum state verification has been introduced \cite{sam2018optimal,zhu2019efficient}. In the verification, one aims to know whether the prepared state $\rho$ is close to the ideal pure state $\ket{\psi}$ in some precision $\epsilon$ for a given significance level $\delta$. The verification is accomplished by a few rounds of 2-outcome verification tests, which constitute the verification operator $\Omega$. Conditioning on the pass of all the tests, one can lower bound the fidelity within a high precision. The efficiency of the verification is determined by the spectral gap of the operator $\Omega$. Comparing to the direct fidelity estimation protocols \cite{flammia2011direct}, the verification protocol is shown to achieve the same fidelity precision with quadratically fewer number of trials.

Inspired by the quantum state verification \cite{sam2018optimal,zhu2019efficient,zhu2019optimal}, here we propose a general framework of the quantum gate verification. The main idea is to map the gate verification to the verification of corresponding Choi state. We first introduce some prior knowledge on the Choi representation and the gate fidelity in Section \ref{Sec:preliminaries}. Then we provide a general framework of quantum gate verification and give the optimal strategies in Section \ref{Sec:framework}. In Section \ref{Sec:typical}, we focus on some typical quantum gates and discuss about their verification strategies. Especially we show that any single-partite (qubit and qudit) gates and Clifford gates can be efficiently verified. In Section \ref{Sec:property}, we discuss the application of the gate verification in testing the properties of quantum channels, such as the robustness of quantum memory \cite{liu2019resource, yuan2019robustness}. Finally, in Section \ref{Sec:conclusion}, we summarize our work, discuss about the possible future direction, and compare it to recent related works.

\section{Preliminaries} \label{Sec:preliminaries}
In this section we first review some essential properties of quantum channels that is related to our discussion.

\subsection{Choi state representation of quantum channels}\label{SSec:duality}
For a quantum system $A$, denote its Hilbert space as $\mc{H}^A$. The set of linear operations on $A$ is denoted as $\mc{L}(\mc{H}^A)$ and the set of quantum states as $\mc{D}(\mc{H}^A)$. Suppose the systems $A$ and $\bar{A}$ own the same dimension and $\mc{B}_A =\{\ket{j}_{A}\}_{j=0}^{d-1}, \mc{B}_{\bar{A}} = \{\ket{j}_{\bar{A}}\}_{j=0}^{d-1}$ are two orthonormal bases of them. The maximally entangled state (with respect to $\mc{B}_A$ and $\mc{B}_{\bar{A}})$ on systems $A,\bar{A}$ is defined to be
\begin{equation}
\ket{\Phi_+}_{A\bar{A}} = \frac{1}{\sqrt{d}} \sum_{j=0}^{d-1} \ket{jj}_{A\bar{A}}.
\end{equation}
and we denote the density matrix $\Phi^{A\bar{A}}_+:= \ket{\Phi_+}_{A\bar{A}}\bra{\Phi_+}$ for simplicity.

A linear map $\mc{E}^{A\to B}: \mc{L}(\mc{H}^A)\to \mc{L}(\mc{H}^B)$ is a quantum channel if and only if (iff) it is a completely positive and trace-preserving (CPTP) map. Denote $\mc{I}_d$ the $d$-dimension identity map. On account of the state-channel duality, the (normalized) Choi state representation of a quantum linear map is defined to be
\begin{equation}
\Phi^{AB}_{\mc{E}} = (\mc{I}^{A\to A} \otimes \mc{E}^{\bar{A}\to B})(\Phi^{A\bar{A}}_+),
\end{equation}
that is, the output state of the map $\mc{I}^{A\to A} \otimes \mc{E}^{\bar{A}\to B}$ with the maximally entangled state as the input state.

The linear map $\mc{E}^{A\to B}$ is completely positive iff $\Phi^{AB}_{\mc{E}}$ is positive; $\mc{E}^{A\to B}$ is trace preserving iff $\tr_B[\Phi^{AB}_{\mc{E}}] = \mathbb{I}_A/d_A$. In this work, we focus on the case when the output dimension $d_B$ is the same as the input dimension $d_A$. We denote $d:= d_A = d_B$. Meanwhile, we omit the superscript of $\mc{E}^{A\to B}$ standing for the system when no ambiguity occurs. Note that as the channel $\mc{E}$ being an unitary $U$, the Choi state is a maximally entangled (pure) state, and we denote the unitary channel as $\mc{U}(\cdot)=U\cdot U^{\dag}$.

The Choi state encodes all the information of the corresponding quantum channel, and one can also obtain the output of the channel by the following relation,
\begin{equation}\label{Eq:dualityAll}
\mc{E}(\rho) = d \tr_A\left[(\rho^T_A \otimes \mathbb{I}_B) \Phi^{AB}_{\mc{E}} \right].
\end{equation}
The state-channel duality is essential to our work, which indicates that verifying the quantum channel is equivalent to verifying the Choi state. We show in Sec.~\ref{Sec:framework} that many results in the state verification can be applied to the current study.

\subsection{Average gate fidelity and entanglement fidelity}
In this work, we focus on benchmarking the quantum gate, say an unitary $U$ on the Hilbert space $\mc{H}_d$. Due to the unavoidable noise, the actual operation realized in an experiment may be a noisy channel $\mc{E}$. Here we use the average gate fidelity to characterize the difference between the ideal unitary gate $\mc{U}$ and the noisy channel $\mc{E}$.
\begin{equation} \label{eq:avgF}
\begin{aligned}
F(\mc{U},\mc{E}) := &\int d \psi \tr \left[\mc{U}(\psi), \mc{E}(\psi)\right]
\end{aligned}
\end{equation}
where the integration is over all the pure state under Haar measure. The average gate fidelity is widely used in the quantum gate benchmarking experiment.

For the corresponding Choi states, the entanglement fidelity is defined as,
\begin{equation} \label{}
\begin{aligned}
F(\mc{U},\mc{E}) := \tr({\Phi_\mc{U}\Phi_{\mc{E}}})=\bra{\Phi_+}\Phi_\Lambda\ket{\Phi_+}.
\end{aligned}
\end{equation}
In fact, there is a direct relation between the average gate fidelity and the entanglement fidelity,
\begin{equation} \label{}
\begin{aligned}
F_A(\mc{U},\mc{E})=\frac{d F_E(\mc{U},\mc{E})+1}{d+1}.
\end{aligned}
\end{equation}

As a result, one can investigate $F_A(\mc{U},\mc{E})$, a practical figure of merit, with $F(\mc{U},\mc{E})$ which is related to the following theoretical derivation. We denote $r(\mc{U},\mc{E}):= 1-F(\mc{U},\mc{E})$ as the entanglement infidelity, and call it infidelity in the following discussion without ambiguity.

\section{General framework of quantum gate verification}\label{Sec:framework}
In this section, we introduce a general framework of quantum gate verification. We first analyze the performance of verification strategies in non-adversarial scenario in Section \ref{SSec:nonad}. We then discuss the optimal verification protocol in Section \ref{SSec:nonOp}, which can be realized in a quite experiment-friendly way. After that, in Section \ref{SSec:adversarial} we extend the verification task to the adversarial scenario, which can be useful in the quantum communication tasks with untrusted quantum channels.

\subsection{Non-adversarial scenario} \label{SSec:nonad}
We start from the i.i.d. (identical and independent distribution) scenario, where a device named Eve is going to produce $N$ rounds of the same quantum channel $\mc{E}$, which should be the unitary gate $\mc{U}$ in the ideal case.
Similar as the state verification, as a user of the channel Alice would like to verify whether the underlying channel is close to the ideal unitary within some $\epsilon$ using $N$ tests under some significance level $\delta$.

On account of the state-channel duality introduced in Sec.~\ref{SSec:duality}, a natural method is to input maximally entangled state and verify the output Choi state directly. However, from a practical point of view, the verification with the maximally entangled state preparation is consumptive and also not robust to the state preparation error. Therefore, in the following discussion, we adopt the strategy that only employs single-partite input states and measurements without ancillaries, that is, in a prepare $\&$ measurement manner.

During each round, Alice prepares a state $\rho_l$, lets it get through the channel $\mc{E}$, and measures it using $2$-outcome positive operator-valued measurement (POVM) operators $\{E_l, 1-E_l\}$, with $0\leq E_l \leq \mathbb{I}$. The state $\rho_l$ and POVM element $E_l$ satisfy
\begin{equation} \label{eq:rholEl1}
\tr[\mc{U}(\rho_l) E_l] = 1.
\end{equation}
We name the combination $(\rho_l, E_l)$ satisfying Eq.~\eqref{eq:rholEl1} as a verification pair for $\mc{U}$.

In different rounds, Alice may adopt different verification pairs $(\rho_l, E_l)$ for testing. Suppose she chooses the pairs with probability $p_l$. The verification pairs $(\rho_l, E_l)$ as well as the probability $p_l$ together compose a strategy $W:= \{p_l, (\rho_l, E_l)\}_l$. The verification protocol is listed as follows.
\begin{enumerate}
    \item For each trial, Alice randomly chooses a verification pair $(\rho_l, E_l)$ with probability $p_l$ from the strategy $W$.
    \item Alice prepares state $\rho_l$, inputs it to the quantum channel $\mathcal{E}$ to be verified, measures the output state using POVM $\{E_l, 1-E_l\}$, and records the test outcome.
    \item Alice performs the above tests for $N$ times. If all the tests pass, Alice estimates the average gate fidelity $F(\mathcal{E},U)\geq 1-\epsilon$ with a significance level $\delta$.
\end{enumerate}

On account of the state-channel duality in Eq.~\eqref{Eq:dualityAll}, Eq.~\eqref{eq:rholEl1} can be reformulated as
\begin{equation}\label{}
d\tr\left[(\rho_l^T\otimes E_l)\Phi_{\mc{U}}\right] = 1.
\end{equation}
and we define the verification operator being
\begin{equation}\label{Eq:stateVer}
\Omega:= d\sum_{l} p_l (\rho^T_l \otimes E_l).
\end{equation}

From this point of view, the verification scheme of a channel is (mathematically) closed related to the one of a maximally entangled state $\Phi_{\mc{U}}$ \cite{zhu2019optimal}. The operator $\Omega$ from the strategy $W$ is denoted as the corresponding verification operator. However, there are still differences between the maximally entangled state verification and the gate verification:
\begin{enumerate}
    \item In the maximally entangled state verification, the possible noisy objects are bipartite states; while in the gate verification, the possible noisy objects are noisy quantum channels, which puts extra limitations on the Choi states compared with the bipartite states.
    \item In the gate verification, the state is prepared deterministically, and the measurement is decided according to the state preparation. Thus one is restricted to the one-way LOCC strategy, comparing to the former bipartite state analysis \cite{wang2019optimal,li2019efficient,Yu2019Optimal}.
\end{enumerate}


Now we study the performance of the verification protocol, which is usually characterized by the minimum number of trials $N(\epsilon,\delta,\Omega)$ for a given infidelity upper bound $\epsilon$, significance level $\delta$, and verification operator $\Omega$. That is, if the verification succeeds in $N$ rounds, one can confirm that the fidelity between the underlying noisy channel and the target unitary is larger than $1-\epsilon$ with probability $1-\delta$.

The minimum number of trials $N(\epsilon,\delta,\Omega)$ is directly related to the maximal passing probability $P(\epsilon,\Omega)$. For the noisy channel with entanglement infidelity $r_E(\mc{U},\mc{E})$ not smaller than $\epsilon$, the maximal pass probability (corresponding to the type-II error of hypothesis testing) is,
\begin{equation}\label{Eq:probE}
\begin{aligned}
P(\epsilon, \Omega) &= \max_{r_E(\mc{U},\mc{E})\geq \epsilon} \tr[\Omega \Phi_\mc{E}] \\
&\leq \max_{\tr[\Phi_{\mc{U}}\rho]\leq 1-\epsilon} \tr[\Omega\rho] = 1 - \nu(\Omega)\epsilon.
\end{aligned}
\end{equation}
Here the first maximization is on all the possible channel $\mc{E}$, and the Choi state should satisfy an additional constraint $\tr_B[\Phi^{AB}_{\mc{E}}] = \mathbb{I}_A/d_A$ than the quantum state verification. Thus the followed inequality acts as an useful upper bound of the pass probability. Here $\nu(\Omega):= 1 - \beta(\Omega)$ is the spectral gap of $\Omega$, with $\beta(\Omega)$ being the second largest eigenvalue. Note that $P(\epsilon,\Omega)$ can be written as a semidefinite program,
\begin{equation} \label{Eq:nonProgram}
\begin{aligned}
\max\quad & \tr\left[\Omega\Phi_\mc{E}^{AB}\right] \\
s.t.\quad & \tr\left[\Phi_\mc{U}^{AB}\Phi_\mc{E}^{AB}\right]\leq 1-\epsilon, \\
& \tr_B\left[\Phi_\mc{E}^{AB}\right] = \frac{\mathbb{I}_d}{d}, \\
& \Phi_\mc{E}^{AB} \geq 0.
\end{aligned}
\end{equation}

Given a verification operator $\Omega$, under the condition of all the $N$ test trials pass, for the significance level $\delta$, i.e., $P(\epsilon,\Omega)^N\leq \delta$, the minimal number of the verification trials $N$ is,
\begin{equation}\label{Eq:upTrial}
\begin{aligned}
N(\epsilon, \delta, \Omega) &= \left\lceil \frac{\ln \delta^{-1}}{\ln P(\epsilon,\Omega)^{-1}}\right\rceil \\
&\leq \left\lceil \frac{\ln \delta^{-1}}{\ln [1 - \nu(\Omega)\epsilon]^{-1}}\right\rceil \leq [\nu(\Omega)\epsilon]^{-1}\ln \delta^{-1},
\end{aligned}
\end{equation}
Here the first inequality is due to the upper bound in Eq.~\eqref{Eq:probE}, which is generally loose.

To reduce the trial number, one should minimize the passing probability in Eq.~\eqref{Eq:probE} for all possible verification operator, and the optimal one is
\begin{equation}\label{Eq:mini}
\begin{aligned}
P^{op}(\epsilon) &=\min_{\Omega}P(\epsilon,\Omega),\\
&=\min_{\Omega}\max_{r_E(\mc{U},\mc{E})\geq \epsilon} \tr(\Omega \Phi_\mc{E}),
\end{aligned}
\end{equation}
where the operator $\Omega$ is from all verification strategy $W$ given by Eq.~\eqref{Eq:stateVer}. The optimal trial number is then $N^{op}(\epsilon, \delta)=\left\lceil \frac{\ln \delta^{-1}}{\ln P^{op}(\epsilon)^{-1}}\right\rceil$. In the following, we show some properties of $P(\epsilon,\Omega)$, which are helpful for its optimization in the next section.
\begin{observation}\label{Ob:convex}
The pass probability $P(\epsilon,\Omega)$ defined in Eq.~\eqref{Eq:probE} is a non-decreasing convex function on the verification operator $\Omega$. That is, $P(\epsilon,\Omega') \geq P(\epsilon,\Omega)$ if $\Omega'-\Omega\geq 0$ is semidefinite positive, and
\begin{equation}\label{Eq:probConv}
\begin{aligned}
P(\epsilon,\Omega') \leq p_1P(\Omega_1,\epsilon)+ p_2P(\Omega_2,\epsilon),
\end{aligned}
\end{equation}
with $\Omega'=p_1\Omega_1+p_2\Omega_2$, $p_1+p_2=1,\ p_1,p_2\geq 0$.
\end{observation}

In practice, the noisy channels $\{\mc{E}_k\}$ during different trials may be different with each other. In this case, a well-defined estimation value would be the averaged infidelity over different rounds
\begin{equation}
\bar{r}(\mc{U},\{\mc{E}_k\}) = \frac{1}{N} \sum_{k=1}^{N} r(\mc{U},\mc{E}_k).
\end{equation}
Similar to the discussion of the quantum state verification \cite{zhu2019general}, with the same verification schemes $W$, one can actually bound the average infidelity $\bar{r}(\mc{U},\{\mc{E}_k\})$ using Eq.~\eqref{Eq:upTrial}.


\subsection{Optimal verification with pure state inputs and projective measurements}\label{SSec:nonOp}

In this section, we provide the optimal verification of any unitary channel $\mc{U}$ under pure state inputs and project measurements (PVM), which is easier for the experiment realization. Suppose there is a verification strategy $W:= \{p_l, (\rho_l, E_l)\}_l$ for the identity channel $\mc{I}$, then any unitary $\mc{U}$ can be verified with $W':= \{p_l, (\rho_l, \mc{U}(E_l))\}_l$. Consequently, without loss of generality we focus on the optimal verification of $\mc{I}$ in the following discussion.

To find the optimal verification of $\mc{I}$, we have the following two lemmas to convert an arbitrary verification operator $\Omega$ to the corresponding Bell-diagonal form without reducing its preformance.
\begin{lemma}\label{Lm:Uinvary}
Under the unitary transformation $\mc{V}$, the verification strategy $W:= \{p_l, (\rho_l, E_l)\}_l$ of the identity channel $\mc{I}$ becomes
$W':= \{p_l, (\mc{V}(\rho_l), \mc{V}(E_l)\}_l$. The pass probability is invariant under the transformation
\begin{equation}\label{Eq:probOmegaprime}
\begin{aligned}
P(\epsilon,\Omega')=P(\epsilon,\Omega),
\end{aligned}
\end{equation}
where the verification operators $\Omega$ and $\Omega'$ are from $W$ and $W'$ respectively and
\begin{equation}\label{Eq:stateVer2}
\begin{aligned}
\Omega'&= d\sum_{l} p_l (\mc{V}(\rho_l)^T \otimes \mc{V}(E_l)) \\
&= d\sum_{l} p_l \mc{V}^*(\rho_l^*) \otimes \mc{V}(E_l)=\mc{V}^*\otimes \mc{V} (\Omega).
\end{aligned}
\end{equation}
\end{lemma}
\begin{proof}
First, note that $\tr[\Omega'\Phi_+]=\tr[\Omega [\mc{V}^*\otimes \mc{V}]^{\dag}(\Phi_+)]=\tr[\Omega \Phi_+]$=1, thus $\Phi_+$ can pass the verification also for $\Omega'$.
Suppose a state $\Phi_{\mc{E}}$ reaches the maximal value of $P(\epsilon,\Omega)$ according to Eq.~\eqref{Eq:probE}, then one can find $\Phi'_{\mc{E}}=\mc{V}^*\otimes \mc{V} (\Phi_{\mc{E}})$ such that $\tr[\Omega' \Phi'_\mc{E}]=\tr[\Omega \Phi_\mc{E}]$. As a result, $P(\epsilon,\Omega')\geq P(\epsilon,\Omega)$. Since the unitary is reversible, similarly one can also get that $P(\epsilon,\Omega')\leq P(\epsilon,\Omega)$, and thus $P(\epsilon,\Omega')=P(\epsilon,\Omega)$.
\end{proof}

\begin{lemma}\label{Lm:Belldiag}
For a verification operator $\Omega$ of the identity channel $\mc{I}$, one can find the corresponding Bell-diagonal verification operator
\begin{equation} \label{Eq:bellT}
\begin{aligned}
\Omega' &= \frac1{d^2}\sum_{u,v=0}^{d-1}\mc{W}^*(u,v)\otimes \mc{W}(u,v) (\Omega)\\
&=\sum_{u,v=0}^{d-1} \lambda_{u,v} \Phi_{u,v},
\end{aligned}
\end{equation}
where $\mc{W}(u,v)$ labeled by $u,v$ are $d^2$ unitary channels of the Weyl operator introduced in Appendix \ref{Sec:quditBell},
such that the pass probability does not increase, i.e., $P(\epsilon,\Omega')\leq P(\epsilon,\Omega)$.
\end{lemma}

The proof of Lemma \ref{Lm:Belldiag} is in Appendix~\ref{Sec:proofs2}.

\begin{theorem}\label{Th:optimal}
For any unitary $\mc{U}$ on $\mc{H}_d$, one can construct the optimal verification strategy with pure state inputs and projective measurements. The optimal verification operator is
\begin{equation}\label{Eq:OpOp}
\begin{aligned}
\Omega_{op} =\frac{\mathbb{I}+d\Phi_{\mc{U}}}{1+d}.
\end{aligned}
\end{equation}
and the optimal pass probability and trial number are
\begin{equation}\label{Eq:Op2}
\begin{aligned}
P^{op}(\epsilon) &= 1-\frac{d}{d+1}\epsilon, \\
N^{op}(\epsilon, \delta) &= \left\lceil \frac{\ln \delta^{-1}}{\ln \left(1 - \frac{d}{d+1}\epsilon\right)^{-1}}\right\rceil \leq \left\lceil \frac{d+1}{d\epsilon}\ln \delta^{-1}\right\rceil.
\end{aligned}
\end{equation}
\end{theorem}

\begin{proof}
Without loss of generality, we consider the identity channel $\mc{I}$ here. Based on Lemma \ref{Lm:Belldiag}, to find the optimal verification one only needs to investigate $\Omega$ in the Bell-diagonal form. In this case, the channel verification and the state verification become coincident, that is, the first inequality in Eq.~\eqref{Eq:probE} is saturated. To be specific, the maximization of $\tr[\Omega\rho]=\tr[\Omega\rho_{\mathrm{diag}}]$ is equivalent for the Bell-diagonal states, which are legal Choi states.

At the same time, for the state verification, the optimal verification operator with separable measurements \cite{Hayashi2006LOCC,zhu2019optimal} is
\begin{equation}\label{Eq:stateOmegaOp}
\begin{aligned}
\Omega_{op} =\frac{\mathbb{I}+d\Phi_+}{1+d},
\end{aligned}
\end{equation}
which is clearly Bell-diagonal, thus can be reached by quantum channel verification. It is clear that the optimal gap here is $\nu(\Omega_{op})=\frac{d}{d+1}$.

Now we show that $\Omega_{op}$ can be constructed in a preparation and measurement manner. The optimal operator $\Omega_{op}$ can be realized by the so called conjugate-basis (CB) projector of an orthogonal basis $\mathcal{B}=\{\psi^{d-1}_{i=0}\}$ in $\mc{H}_d$ \cite{zhu2019optimal},
\begin{equation}
\begin{aligned}
P(\mathcal{B})= \sum_{\psi_l\in\mathcal{B}}\psi^*_l\otimes \psi_l.
\end{aligned}
\end{equation}
That is, $\Omega_{op}=\frac1{d+1}\sum_{l=1}^{d+1} P(\mathcal{B}_l)$, when ${B}_l$ are $d+1$ mutually unbiased bases (MUBs). If the dimension is not a prime power, the verification operator can be realized by $\Omega_{op}=\sum_{\alpha}p_{\alpha} \phi^*_{\alpha}\otimes \phi_{\alpha}$ and $\sum_{\alpha} p_{\alpha}=d$, with the weighted complex projective 2-design $\{p_{\alpha},\ \phi^*_{\alpha}\}$ \cite{zhu2019optimal,ZAUNER2011DESIGNS,Renes2004Symmetric}.

Finally, according to  Eq.~\eqref{Eq:dualityAll}, the corresponding verification strategy of $\mc{I}$ shows, $\{\frac1{d(d+1)},\ (\psi^i_l, \psi^i_l)\}$, where $\psi^i_l$ is from $(d+1)$ MUB ${B}_l$. That is, we input $\psi^i_l$ and measure $\psi^i_l$ with equal probability. For the unitary $\mc{U}$, the verification strategy is  $\{\frac1{d(d+1)},\ [\psi^i_l, \mc{U}(\psi^i_l)]\}$. One can find the strategy of $\Omega$ constructed from $2$-designs in a similar manner.
\end{proof}

Practically, one may prefer to implement the verification with less MUBs, due to the reasons that there are no enough MUBs in the Hilbert space or to reduce the experiment resources. The verification can be built with less MUBs, $\Omega=\frac1{g}\sum_{l=1}^{g} P(\mathcal{B}_l)$, and the spectral gap is $\nu(\Omega)=(g-1)/g$ \cite{zhu2019optimal}. According to Eq.~\eqref{Eq:upTrial}, the trial number is upper bounded by,
\begin{equation}\label{}
\begin{aligned}
N(\epsilon, \delta)\leq \left\lceil \frac{\ln \delta^{-1}}{\ln \left(1 - \frac{g-1}{g}\epsilon\right)^{-1}}\right\rceil \leq \left\lceil\frac{g}{(g-1)\epsilon}\ln \delta^{-1}\right\rceil,
\end{aligned}
\end{equation}
Note that the bound may be not tight, however it is economical. For example, one can finish the verification with only two bases with the trial number only about two times overhead than the optimal one.

\subsection{Adversarial scenario}\label{SSec:adversarial}

In the discussion above, we suppose the implemented quantum gates are independent for different rounds. However, this may not be true in general. In some practical quantum information tasks, the quantum channels in different rounds will be correlated. For example, when Alice produces the uncharacterized gates with memory effect, the gate noise in the former rounds may affect the latter gate realization. On the other hand, in some quantum communication tasks, the quantum channels may be held by some untrusted parties Eve, e.g. entangled state distribution and quantum key distribution \cite{lo1999unconditional}. In this case, the adversarial Eve may be even more powerful so that he can take advantage of the correlations between different rounds \cite{zhu2019general}. Eve may produce a large composite quantum channel
\begin{equation}
\mc{E}_{(N+1)} : \mc{D}((\mc{H}^A)^{\otimes (N+1)}) \to \mc{D}((\mc{H}^B)^{\otimes (N+1)}).
\end{equation}
with arbitrarily correlated noise. We will leave out the subscript $(N+1)$ in the later discussion in this section, i.e., $\mc{E}:=\mc{E}_{(N+1)}$.

To verify the quantum channel in this case, we suppose Alice (and Bob) is able to randomly choose $N$ rounds from the overall $(N+1)$ rounds to perform the verification test. She (They) leaves the left round to perform the real quantum information processing task. In Fig.~\ref{fig:ChanVer}(b), we describe the adversarial channel verification with two parties.

The possibility that the $N$ rounds of tests pass is
\begin{equation} \label{Eq:pmcE}
p_{\mc{E}} = \tr\left[(\Omega^{\otimes N}\otimes I)\Phi_{\mc{E}_{(N+1)}}\right],
\end{equation}
where without loss of generality, we assume the test is on the first $N$ qubits, and in the same time $\Phi_{\mc{E}_{(N+1)}}$ is permutation-invariant.
Conditioning on the passing of $N$ rounds tests, Alice would like to confirm that the reduced $(N+1)$-th round quantum channel given by the reduced Choi state
\begin{equation}
\Phi_{\mc{E}'} = p^{-1}_{\mc{E}}\tr_{1\sim N}\left[(\Omega^{\otimes N}\otimes I)\Phi_{\mc{E}_{(N+1)}}\right],
\end{equation}
is closed to the target unitary $U$. The entanglement fidelity between $\Phi_{\mc{E}'}$ and $U$ is
\begin{equation}
\begin{aligned}
F(\mc{E}', \mc{U}) &= \tr(\Phi_{\mc{E}'}\Phi_\mc{U}) \\
&= p^{-1}_{\mc{E}} \tr\left[(\Omega^{\otimes N}\otimes \Phi_{\mc{U}})\Phi_{\mc{E}_{(N+1)}}\right] = p^{-1}_{\mc{E}} f_{\mc{E}},
\end{aligned}
\end{equation}
where
\begin{equation} \label{Eq:fmcE}
f_{\mc{E}}:= \tr\left[(\Omega^{\otimes N}\otimes \Phi_{\mc{U}})\Phi_{\mc{E}_{(N+1)}}\right].
\end{equation}

\begin{figure*}[htbp]
\includegraphics[width=16cm]{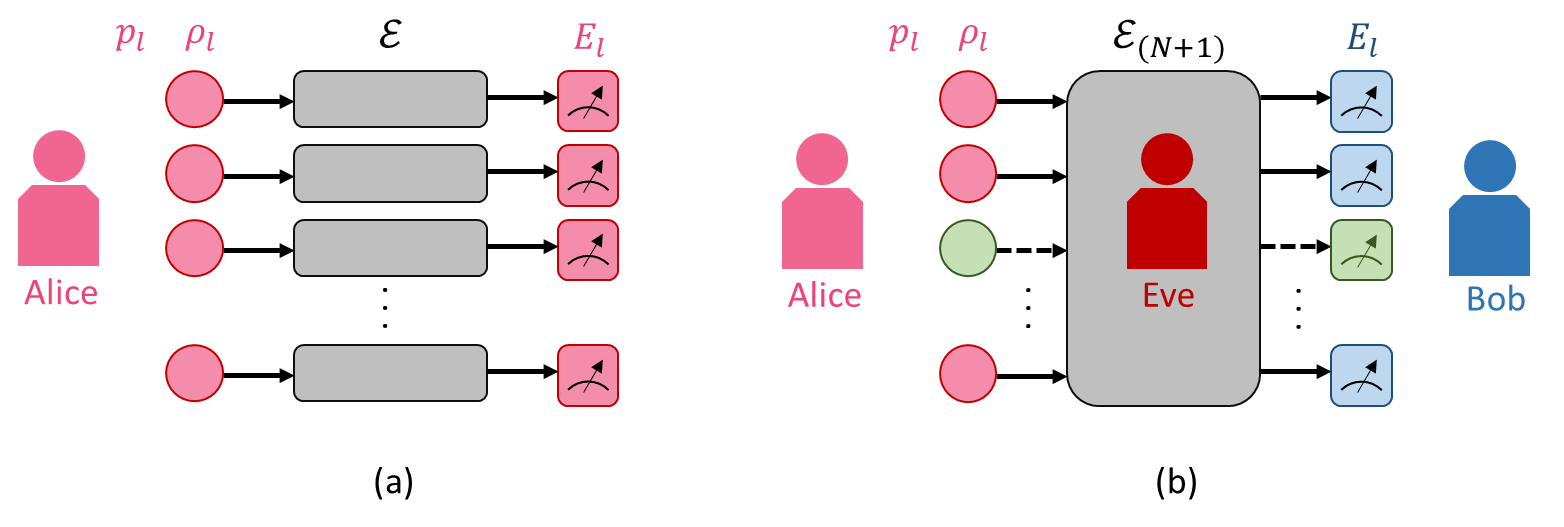}
\caption{The non-adversarial scenario and adversarial scenario. (a) In the non-adversarial scenario, Alice prepare the states $\rho_l$, sends it to an uncharacterized channel and performs measurement $E_l$ on it. The channels of different trials are independent with each other. (b) In the adversarial scenario with two communication parties, Alice prepare the states $\rho_l$, sends it to an untrusted channel, Bob then receives output states from the channel. After Alice announces the random test rounds, Bob performs measurement $E_l$ on them and estimate the gate for the left turn (shown in green). The channels of different trials are correlated with each other.} \label{fig:ChanVer}
\end{figure*}

The core task in adversarial scenario is to verify whether the channel used for the task round is the target unitary channel $\mc{U}$. Similarly to the state verification discussion in Ref.~\cite{zhu2019general}, we define the estimated (entanglement) fidelity lower bound with respect to the number of test rounds $N$, a failure probability of $\delta$, and the verification strategy $\Omega$
\begin{equation}
F(N,\delta,\Omega) := \min_{\Phi_{\mc{E}}}\{p^{-1}_{\mc{E}} f_{\mc{E}} | p_{\mc{E}} \geq \delta\}, \quad 0<\delta<1,
\end{equation}
where $\Phi_{\mc{E}}$ take values over all Choi states. The number of trials lower bound with respect to a precision of $\epsilon$, a failure probability of $\delta$, and the verification strategy $\Omega$ is defined to be
\begin{equation}
N(\epsilon,\delta,\Omega) := \min\{N| F(N,\delta,\Omega)\geq 1-\epsilon \}.
\end{equation}

For convenience of the later discussion, we also define the bipartite \emph{state} verification parameters
\begin{equation}
\begin{aligned}
F_S(N,\delta,\Omega) &:= \min_{\rho} \{p^{-1}_{\rho} f_{\rho} | p_{\rho} \geq \delta\}, \\
N_S(\epsilon,\delta,\Omega) &:= \min\{N| F_S(N,\delta,\Omega)\geq 1-\epsilon \}.\\
\end{aligned}
\end{equation}
Here the optimization is taken over all the $2(N+1)$-qudit $(\bigotimes_{i=1}^{N+1}\mc{H}_i)^{\otimes2}$ bipartite state $\rho$, and $p_{\rho}, f_{\rho}$ is defined by replacing $\Phi_\mc{E}$ in Eq.~\eqref{Eq:pmcE} and Eq.~\eqref{Eq:fmcE} to $\rho$. It is obvious that $F(N,\delta,\Omega)\geq F_S(N,\delta,\Omega)$ and $N(\epsilon,\delta,\Omega)\leq N_S(\epsilon,\delta,\Omega)$. Therefore, the bipartite state verification parameter $F_S(N,\delta,\Omega)$ and $N_S(\epsilon,\delta,\Omega)$ are the lower bound and upper bound of $F(N,\delta,\Omega)$ and $N(\epsilon,\delta,\Omega)$ respectively. One can apply the analysis in Ref.~\cite{zhu2019efficient,zhu2019general} to estimate $N_S(\epsilon,\delta,\Omega)$ and $F_S(N,\delta,\Omega)$, which provides a useful bound for $N(\epsilon,\delta,\Omega)$ and $F(N,\delta,\Omega)$.

For a general strategy $\Omega$, $F(N,\delta,\Omega)$ can be expressed as the following programming problem
\begin{equation}
\begin{aligned}
\min \quad & \tr\left[(\Omega^{\otimes N}\otimes \Phi_\mc{U})\Phi_{\mc{E}}\right]/\tr\left[(\Omega^{\otimes N}\otimes I)\Phi_{\mc{E}}\right] \\
s.t. \quad & \tr\left[(\Omega^{\otimes N}\otimes I)\Phi_{\mc{E}}\right] \geq \delta \\
& \tr_B\left[\Phi_\mc{E}\right] = \left(\frac{\mathbb{I}_d}{d}\right)^{\otimes (N+1)} \\
& \Phi_\mc{E} \geq 0,
\end{aligned}
\end{equation}
which is not easy to find an analytical solution in general.

In the following paragraphs, we show the method to find the optimal verification schemes as well as to analyze its performance. We first consider the verification operator in the Bell-diagonal form, and show that the figure of merits equal to the ones of the state. Then, we extend the analysis to a general type of verification operators which are called Bell-supported, and show that they are always sub-optimal to a homogeneous strategy. Finally, we solve the optimal homogeneous strategy and the performance of it.

\begin{observation}\label{Ob:adBelldiag}
For a verification strategy $\Omega$ of a quantum gate $U$ which is bell-diagonal under a local unitary transformation, i.e.,
\begin{equation} \label{Eq:adbelldiagonal}
\Omega = \sum_{u,v=0}^{d-1} \lambda_{u,v} \tilde{\Phi}_{u,v}^{AB},
\end{equation}
where $\{\tilde{\Phi}_{u,v}^{AB}\}$ are the qudit Bell states $\{\Phi_{u,v}^{AB}\}$ under local unitary transformation on system $A$ and $B$, and $\tilde{\Phi}_{0,0}^{AB} = \Phi_{\mc{U}}^{AB}, \lambda_{0,0}=1$, we have
\begin{equation}
\begin{aligned}
F(N,\delta,\Omega) = F_S(N,\delta,\Omega), \\
N(\epsilon,\delta,\Omega) = N_S(\epsilon,\delta,\Omega). \\
\end{aligned}
\end{equation}

\end{observation}

\begin{proof}
We first simplify the expression of $F(N,\delta,\Omega)$. Due to the random assignment of test rounds, it is not restrictive to consider the permutation-invariant $\Phi_{\mc{E}}$ only. Similar to the discussion in Ref.~\cite{zhu2019general}, one can define the permutation-invariant Bell basis
\begin{equation}
\tilde{\Phi}_{\mb{k}} = \hat{\mb{P}}_S(\tilde{\Phi}_{0,0}^{\otimes k_{0,0}}\otimes \tilde{\Phi}_{0,1}^{\otimes k_{0,1}} \otimes ...\otimes \tilde{\Phi}_{d-1,d-1}^{\otimes k_{d-1,d-1}}),
\end{equation}
where $\hat{\mb{P}}_S$ is the symmetrization operator, mixing all possible permutation with respect to different rounds, $k:=[k_{0,0}, k_{0,1}, ..., k_{d-1,d-1}]$ is a sequence of nonnegative integer number with $\sum_{u,v} k_{u,v} = N+1$.

Since $p_{\mc{E}}$ and $f_{\mc{E}}$ in Eqs.~\eqref{Eq:pmcE},~\eqref{Eq:fmcE} only depend on the diagonal elements of $\Phi_{\mc{E}}$ in the Bell basis, without loss of generality, we may assume that the Choi state is diagonal in the product basis of $\tilde{\Phi}_{u,v}$. We only need to consider the Choi state $\Phi_{\mc{E}}$ as the mixture of $\tilde{\Phi}_{\mb{k}}$
\begin{equation}
\Phi_{\mc{E}} = \sum_{\mb{k}\in\mb{K}} c_{\mb{k}} \tilde{\Phi}_{\mb{k}},
\end{equation}
where $\{c_{\mb{k}}\}$ are the nonnegative mixing coefficients with $\sum_{k\in\mb{K}} c_{\mb{k}} = 1$, and $\mb{K}$ is the set of all possible $\mb{k}$. Note that, the $\tilde{\Phi}_{u,v}$-basis naturally meets the requirements of Choi states, i.e., $\tr_B[\Phi_{u,v}^{AB}] = I_d/d$. As a result, the optimization is over the whole convex hull made by $\{\Phi_{\mb{k}}\}$, similar to the state case in Ref.~\cite{zhu2019general}. Therefore,
\begin{equation}
\begin{aligned}
F(N,\delta,\Omega) &= \min_{\Phi_{\mc{E}}}\{p_{\mc{E}}^{-1}f_{\mc{E}}|p_{\mc{E}}\geq \delta\} \\
&= \min_{\{c_{\mb{k}}\}}\{ p_{\mc{E}}^{-1}f_{\mc{E}}|p_{\mc{E}}\geq \delta \} \\
&= F_S(N,\delta,\Omega).
\end{aligned}
\end{equation}
\end{proof}

A strategy $\Omega$ for unitary $\mc{U}$ with the form
\begin{equation}
\Omega = \Phi_{\mc{U}} + \lambda(1 - \Phi_{\mc{U}}), \quad (0\leq \lambda<1),
\end{equation}
is called homogeneous. Note that the homogeneous strategy is a specific case of bell-diagonal strategy. The eigenvalues of such $\Omega$ except the largest one are all degenerated to be $\lambda$. It was shown in Ref.~\cite{zhu2019general} that the following optimization of the quantum state verification
\begin{equation}
\max_{\Omega} F_S(N, \delta, \Omega)
\end{equation}
can always be achieved by the homogeneous strategy for given $N$ and $\delta$.

Now we discuss the optimal strategy $\Omega$ for the quantum gate verification and first introduce some notations. We call a strategy $\Omega$ \emph{useless} under given $N$ and $\delta$ if no Choi state $\Phi_{\mc{E}}$ meets the requirement
\begin{equation}
p_{\mc{E}} \geq \delta.
\end{equation}

By spectrum decomposition, a strategy $\Omega$ can be written in the following unique form
\begin{equation}
\Omega = \sum_{j=0}^{J-1} \lambda_j \Pi_j,
\end{equation}
where $J<d$ is the number of different eigenvalues, $\lambda_0 = 1 > \lambda_1 > ... > \lambda_{J-1}\geq 0$, and $\Pi_j$ is the projector onto the eigenspace with eigenvalue $\lambda_j$, whose rank may be larger than 1. If there exists a maximally entangled state $\Phi_{e}$ such that $\Phi_{e}\subseteq\Pi_j$, we call the $\Pi_j$ space is \emph{Bell-supported}. Denote the set of Bell-supported $\{\Pi_j\}$ of $\Omega$ as $\mb{S}(\Omega)$. Obviously, $\Pi_0\subseteq \mb{S}(\Omega)$. If a strategy has Bell-supported projector set $\mb{S}(\Omega)$ with at least one elements else than $\Pi_0$, we call the strategy $\Omega$ is Bell-supported. The Bell-diagonal strategies are the extreme cases of Bell-supported strategies, where $\mb{S}(\Omega)$ span the whole operator space of $\Omega$.

For the Bell-supported strategies, we have the following lemma.
\begin{lemma} \label{Lm:Bellsupport}
For a Bell-supported strategy $\Omega$, denote a subset of $\mb{S}(\Omega)$ as $\mb{S}_0(\Omega)\subseteq \mb{S}(\Omega)$ which contains $\Pi_0$ and at least another element $\Pi_j$. Denote the set of eigenvalues corresponding to the projects in $\mb{S}_0(\Omega)$ as $\lambda(\mb{S}_0(\Omega))$. If we construct a new strategy $\Omega'$ with the following form
\begin{equation}
\Omega' = \sum_{j|\Pi_j\in \mb{S}_0(\Omega)} \lambda_j \Pi_j  + \sum_{j|\Pi_j\notin \mb{S}_0(\Omega)} \tilde{\lambda}_j \Pi_j,
\end{equation}
where $\tilde{\lambda}_j$ can take any value in $\lambda(\mb{S}_0(\Omega))$ except for $\lambda_0=1$, then
\begin{equation}
F(N,\delta,\Omega') \geq F(N,\delta,\Omega)
\end{equation}
if $\Omega'$ is not useless given $N$ and $\delta$.
\end{lemma}

Lemma \ref{Lm:Bellsupport} implies that, for the Bell-supported strategy $\Omega$, one can always find a strategy with degenerated eigenvalue which is not worse than $\Omega$. Therefore, for a given $N$ and $\delta$, and among all the Bell-supported strategies $\Omega$, by applying the Lemma \ref{Lm:Bellsupport}, one can see that the optimal strategy can always be achieved by homogeneous strategy. The proof of Lemma \ref{Lm:Bellsupport} is in Appendix \ref{Sec:proofs3}.

For the homogeneous strategy $\Omega$, according to Observation \ref{Ob:adBelldiag}, one can directly calculate $F_S(N,\delta,\Omega)$ and $N_S(\epsilon,\delta,\Omega)$. In the high precision limit, i.e., $\epsilon,\delta\to 0$, the optimal homogeneous strategy to verify $U$ is
\begin{equation}
\Omega = \Phi_\mc{U} + \frac{1}{e} (1 - \Phi_\mc{U}).
\end{equation}

To realize this, based on the optimal CB-test strategy introduced in Section \ref{Sec:framework}, one may add some ``trivial test'' into it. In the ``trivial test'', Alice and Bob perform no operation to realize the identity test. To realize the optimal homogeneous test in the high precision limit, one may perform the trivial test with probability $p=\frac{d+1-e}{ed}$ and original optimal CB-test with probability $1-p$. In this case, the required number of trials is \cite{zhu2019optimal}
\begin{equation}
N(\epsilon,\delta,\lambda) = N_S(\epsilon,\delta,\lambda) \approx e\epsilon^{-1}\ln\delta^{-1}.
\end{equation}

\section{Verification of some typical quantum gates} \label{Sec:typical}
In the previous section, we introduce the general framework of the quantum gate verification. Especially, we show that any unitary channel $\mc{U}$ on $\mc{H}_d$ can be efficiently verified with pure state inputs and projective measurements, in both non-adversarial and adversarial scenario. In this section, we apply such verification protocol to several typical quantum gates involved in quantum computing, such as any single qubit gates, multi-qubit Clifford gates and beyond. Hereafter, we focus on the non-adversarial scenario.

\subsection{Single qubit gates}\label{SubSec:single}
We first study the qubit identity channel $\mc{I}$, and latter directly extend it to any single qubit gate $\mc{U}$ by some unitary transformation. The Choi state of $\mc{I}$ is $\Phi_+$. According to Theorem \ref{Th:optimal}, we can utilize 3 MUBs from the Pauli bases,
\begin{equation}
\begin{aligned}
P(X)&=\frac{X\otimes X + \mathbb{I}}{2}=\ket{++}\bra{++}+\ket{--}\bra{--},\\
P(Y)&=\frac{-Y\otimes Y+\mathbb{I}}{2}=\ket{+i-i}\bra{+i-i}+\ket{-i+i}\bra{-i+i},\\
P(Z)&=\frac{Z\otimes Z+\mathbb{I}}{2}=\ket{00}\bra{00}+\ket{11}\bra{11},
\end{aligned}
\end{equation}
which account for three subspaces, and $\ket{\pm i}$ denote the eigenstates of the $Y$ basis. Note that these three projectors can also be derived from the stabilizer of the Choi state, which is helpful for the derivation of multi-qubit gates. The verification operator is $\Omega=\frac1{3}(P_X+P_Y+P_Z)$ \cite{sam2018optimal}. By Theorem \ref{Th:optimal}, the qubit gate can be verified with optimal trial number $N^{op}(\epsilon, \delta)=\left\lceil \frac{\ln \delta^{-1}}{\ln [1 - \frac{2}{3}\epsilon]^{-1}}\right\rceil \left\lceil \frac{3}{2\epsilon}\ln \delta^{-1}\right\rceil$.

The corresponding verification strategy $W$ for the identity qubit channel $\mc{I}$ is to choose the following verification pairs $(\rho_l, E_l$)
\begin{equation}
\begin{aligned}
&(\ket{+}, \ket{+}),\ (\ket{-}, \ket{-}),\\
&(\ket{+i}, \ket{+i}),\ (\ket{-i}, \ket{-i}),\\
&(\ket{0}, \ket{0}),\ (\ket{1}, \ket{1}).
\end{aligned}
\end{equation}
with equal probability $1/6$. For example, $(\ket{+}, \ket{+})$ means that one inputs the $\ket{+}$ and perform measurement using PVM $\{\ket{+}\bra{+}, \mathbb{I} - \ket{+}\bra{+}\}$. If the mesurement result is $\ket{+}\bra{+}$, the test passes. For any single qubit gate $\mc{U}$, verification pairs should be updated to ($\rho_l, \mc{U}(E_l)$). For example, for the $Z$ gate the verification strategy is to choose
\begin{equation}
\begin{aligned}
&(\ket{+}, \ket{-}),\ (\ket{-}, \ket{+}),\\
&(\ket{+i}, \ket{-i}),\ (\ket{-i}, \ket{+i}),\\
&(\ket{0}, \ket{0}),\ (\ket{1}, \ket{1}).
\end{aligned}
\end{equation}
with equal probability $1/6$. In the same way, the non-Clifford $T$ gate can also be verified. In addition, general qudit gates can be verified according to Sec.~\ref{SSec:nonOp}.

\subsection{Clifford gates}\label{SubSec:Clifford}
In this and the next section, we consider the multi-qubit gates, where the underlying Hilbert space is $\mc{H}_d=\mc{H}_2^{\otimes N}$. In this case, it is not easy to implement the optimal strategy given in Sec.~\ref{SSec:nonOp}, since the input states and the measurements could be entangled ones. Thus, in the following we show how to verify the Clifford and $C^{n-1}Z(x)$ gates locally, inspired by the verification of stabilizer(-like) states.

\begin{figure}[htbp]
\includegraphics[width=3.8cm]{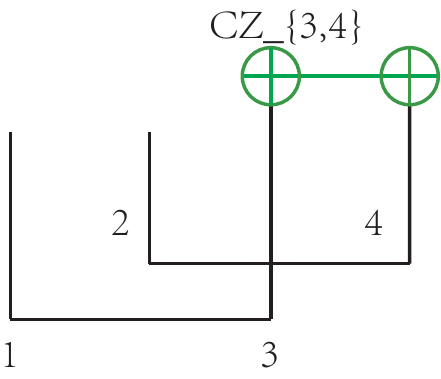}
\caption{The Choi state: the CZ gate operates on the Bell pairs. The green (horizontal) line labels the CZ gate, and the black $U$-type line labels the Bell pair.} \label{fig:CZ}
\end{figure}

Let us first take the Controlled-Z (CZ) gate as an example, and the overall Choi state shows,
\begin{equation}
\begin{aligned}
\ket{\Phi_{CZ}}=\frac1{2}CZ_{3,4}(\ket{00}+\ket{11})_{1,3}\otimes (\ket{00}+\ket{11})_{2,4}.
\end{aligned}
\end{equation}
Note that CZ gate operates on the final two qubits. The stabilizer generators of the initial Bell states are,
\begin{equation}
\begin{aligned}
g_1=X_1X_3,\ g_2=Z_1Z_3,\ g_3=X_2X_4,\ g_4=Z_2Z_4,
\end{aligned}
\end{equation}
and the generators of the state $\ket{\Phi_{CZ}}$ is updated to $g_i'=\mc{U}(g_i)$, where $U$ is the corresponding gate (CZ here).
\begin{equation}
\begin{aligned}
g'_1=X_1X_3Z_4,\ g'_2=Z_1Z_3,\ g'_3=X_2Z_3X_4,\ g'_4=Z_2Z_4,
\end{aligned}
\end{equation}
on account the commuting relations,
\begin{equation}\label{Eq:commCl}
\begin{aligned}
&CZ_{i,j}X_{i(j)}CZ_{i,j}=X_iZ_j(Z_iX_j),\\
&CZ_{i,j}Z_{i(j)}CZ_{i,j}=Z_{i(j)}.
\end{aligned}
\end{equation}

To verify the Choi state, we can use the four stabilizer generators $g_i'$ to construct the projection $P_i=\frac{g_i'+\mathbb{I}}{2}$, and the verification operator is $\Omega = \frac{1}{2n}\sum_i P_i$ (here $n=2$) with the gap being $\nu(\Omega)=1/2n$. In fact, one can utilize all the non-trivial $2^{2n}-1$ stabilizers to enhance the gap to $2^{2n-1}/(2^{2n}-1)$ \cite{sam2018optimal,zhu2019general}, but may cost more state preparation and measurement settings. In some cases, the measurement settings can be reduced by the coloring of the corresponding graph states \cite{Toth2005Detecting,Zhu2018hypergraph,Zhou2019structure}, which is equivalent to the stabilizer states under local Clifford gates \cite{Hein2006Graph}. For example, for the $n$-qubit Clifford circuit with only CZ gates between each two neighbouring qubits, the corresponding Choi state is a $2$-color graph state. Therefore, with only two state preparation and measurement settings, one can verify such Clifford circuit efficiently.

Now we translate the strategy expressed by verification operator $\Omega$ to the realization with verification pairs $(\rho_l, E_l)$. For the projector $P_i$, the corresponding subspace is the $+1$ subspace of $g_i'=A_i\otimes B_i$, where $A_i, B_i$ are two Pauli tensor operators. Thus the verification strategy  $(\rho_l, E_l)$ is to input the eigenstate $\ket{\psi_A}$ in the $+1(-1)$ subspace and project the eigenstate to the  $+1(-1)$ subspace of $B_i$. Since $A_i, B_i$ are Pauli operators, the verification can be accomplished with inputting product pure states in the Pauli basis and Pauli measurements. For instance, the verification pairs of projector $P_1$ and $P_2$ are,
\begin{equation}
\begin{aligned}
&\{\ket{+}_{1},\ (X_3Z_4)^{+}\},\ \{\ket{-}_{1},\ (X_3Z_4)^{-}\},\\
&\{\ket{0}_{1},\ Z_3^{+}\},\ \{\ket{1}_{1},\ Z_3^{-}\},\\
\end{aligned}
\end{equation}
similarly for $P_2$ to $P_4$. To be specific, here $\{\ket{+}_{1},\ (X_3Z_4)^{+}\}$ means that one inputs $\ket{+}$ on the first qubit ($\mathbb{I}/2$ on the second qubit),
and measure the result in the $+1$ basis of $X_3Z_4$. It is clear that $(X_3Z_4)^{\pm}$ can be finished by local $X_3$ and $Z_4$ measurements and classical post-processings.

The above analysis can be generalized to the verification of any Clifford gates, and we summarize this in the following observation.
\begin{observation}\label{Ob:Clifford}
Any n-qubit Clifford gate can be verified under entanglement infidelity $\epsilon$ and significance level $\delta$ with verification trial number upper bounded by,
\begin{equation}\label{}
\begin{aligned}
N\leq \left\lceil\frac{2n}{\epsilon}\ln\delta^{-1}\right\rceil
\end{aligned}
\end{equation}
and this bound can be further reduced with more input states and measurement settings,
\begin{equation}\label{Eq:clif}
\begin{aligned}
N\leq \left\lceil\frac{2^{2n}-1}{2^{2n-1}}\epsilon^{-1} \ln\delta^{-1}\right\rceil
\end{aligned}
\end{equation}
where the input states are in the Pauli basis and the measurements are local Pauli ones.
\end{observation}

\subsection{Multi-qubit Control-Z and Control-X gates}
In this section, we show the verification protocol of the $C^{n-1}Z$ and $C^{n-1}Z$, i.e.,
\begin{equation}
\begin{aligned}
C^{n-1}Z=\mathbb{I}-2\ket{00\cdots0}\bra{00\cdots0}.
\end{aligned}
\end{equation}
and $C^{n-1}X=H_nC^{n-1}XH_n$.

\begin{figure*}[htbp]
\includegraphics[width=13cm]{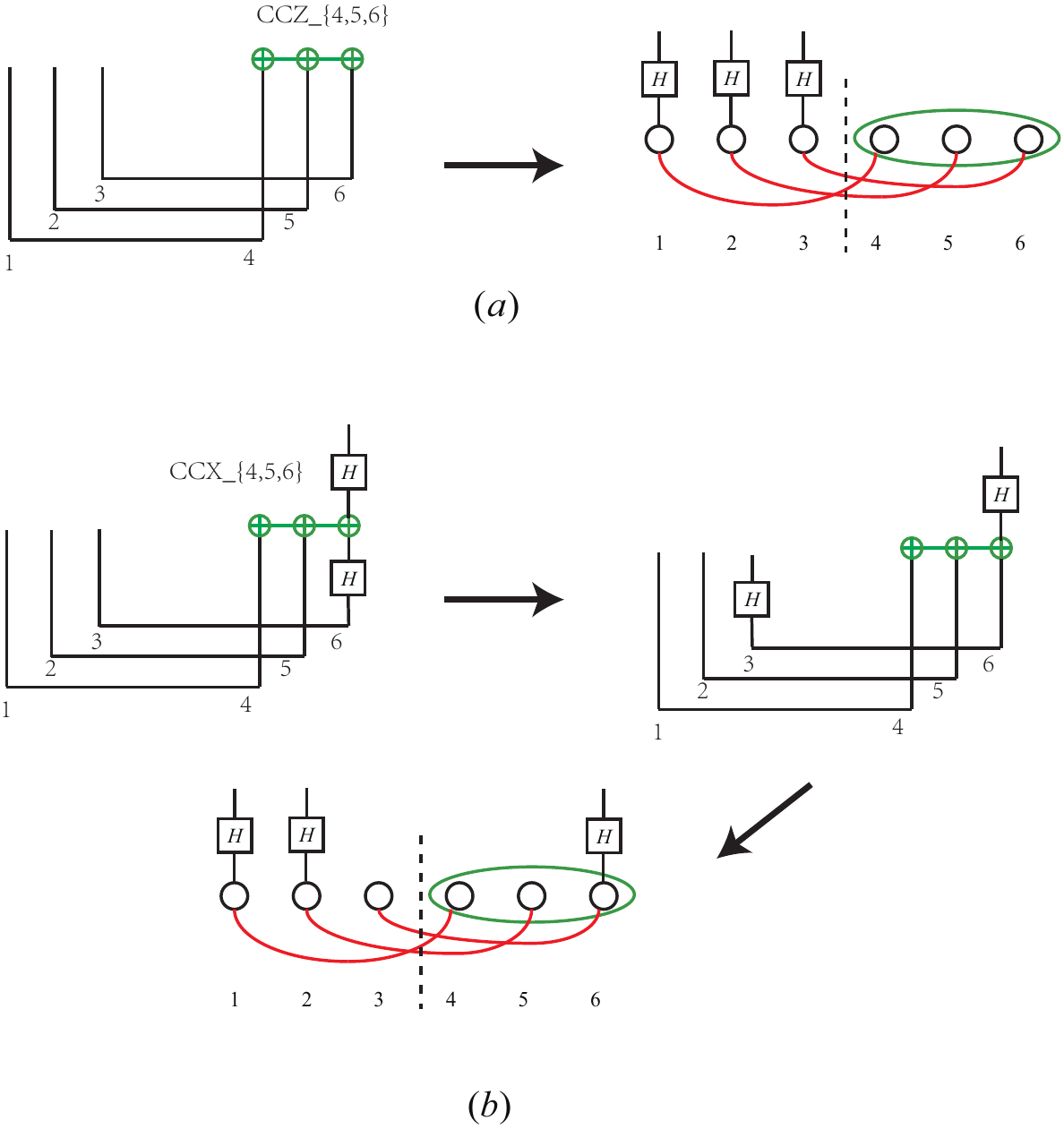}
\caption{The Choi state: the CCZ gate operates on the Bell pairs. The green (horizontal) line labels the CCZ gate, and the black $U$-type line labels the Bell pair. Here we transform the Choi states of $C^{n-1}Z$ and $C^{n-1}X$ to the hypergraph states $\ket{HG}$ in (a) and (b) respectively. Here the hypergraph state $\ket{HG}$ owns three (red) normal edge and one (green) $n=3$-hyperedge, and the graph is $n+1=4$ colorable.}\label{fig:CCZ}
\end{figure*}

Similar as Sec.~\ref{SubSec:Clifford}, we can find the updated ``stabilizer" generators, however now the stabilizers are not in the Pauli tensor form since the $C^{n-1}Z(x)$ gate is not a Clifford one.
Since the $C^{n-1}Z$ gate can generate the hypergraph state \cite{Rossi2013hyper}, in the following we adopt the verification operator for the hypergraph state \cite{Zhu2018hypergraph}. As shown in Fig.~\ref{fig:CCZ}, the Choi state of the $C^{n-1}Z(X)$ state is equivalent to the hypergraph state under local unitary, i.e, the single-qubit Hadamard gate.

\begin{equation}
\begin{aligned}
&\ket{\Phi_{C^{n-1}Z}}=\bigotimes_{i=1}^n H_i \ket{HG},\\
&\ket{\Phi_{C^{n-1}X}}=\bigotimes_{i=1}^{n-1}H_i \otimes H_{2n} \ket{HG},\\
\end{aligned}
\end{equation}
where the hypergraph state $\ket{HG}=\bigotimes_{i=1}^n CZ_{\{i,i+n\}}C^{n-1}Z\ket{+}^{\otimes N}$.

In this way, we can directly get the stabilizer generators of the Choi states from the ones of the hypergraph state. For example, for the $\ket{\Phi_{C^{n-1}Z}}$, the generator related to the 4-th qubit is $g_4=X_1X_4CZ_{5,6}$. It is not hard to see that the graph is $n+1$ colorable in Fig.~\ref{fig:CCZ}. Thus one can verify $\ket{\Phi_{C^{n-1}Z}}$ and $\ket{\Phi_{C^{n-1}X}}$ using verification operator constructed from the stabilizers with the spectral gap is $1/(n+1)$ according to the color protocol in \cite{Zhu2018hypergraph}. In a similar way as in Sec.~\ref{SubSec:Clifford}, we can transfer the state protocol to the verification strategy of the unitary gates, and the verification trial number is $N\leq \left\lceil\frac{n+1}{\epsilon}\ln \delta^{-1}\right\rceil$. Note that we can still use local state inputs and Pauli basis measurements, since the $C^{k}Z$ operator on k-qubit can be measured with the $Z$ basis measurement $Z^{\otimes k}$ using post-processing.

\section{Applications in channel property testing} \label{Sec:property}
In this section, we show the application of the verification protocol to the efficient test on the property of the underlying quantum channel. Here we focus on channels' entanglement property. We believe that the following analysis could be easily generalized to other properties, such as the coherence generating power \cite{Mani2015Cohering,Diaz2018usingreusing}.

The entanglement property refers to whether the underlying channel is an entanglement-preserving (EP) or the entanglement-breaking (EB) one. This kind of test is essential for quantum communications, such as the quantum memory and the quantum channel in quantum networks and distributed quantum computing. An EB channel can always be described by a measurement-and-preparation channel, thus destroys any quantum correlation between the initial input state and other parties. In the following sections, we first discuss the verification of the entanglement property of the channel and further quantifies this kind of quantumness by estimating a lower bound of the (generalized) robustness of the quantum memory \cite{liu2019resource,yuan2019robustness}.

\subsection{Entanglement property detection}
As a specific type of quantum channel, a good quantum memory can preserve the quantum information to some extent. In the ideal case, the quantum memory keeps all the information contained in the states and is reversible. The perfect memory is a known unitary $\mc{U}$, e.g., the identity channel $\mc{I}$. In the following discussion, we show that the verification protocol can help us reveal whether the noisy channel is EP. Without loss of generality, here we focus on the strategies to verify $\mc{I}$.

It is known that a channel is EP iff the corresponding Choi state is an entangled state. The Choi state $\Phi_{\mc{E}}$ is entangled if the fidelity to the maximally entangled state $\mathrm{Tr}(\Phi_{\mc{E}}\Phi^+)\geq 1/d$, that is, by the violation of the following witness,
\begin{equation}\label{Eq:witness}
\begin{aligned}
\mc{W}:=\frac{\mathbb{I}}{d}-\Phi^+,
\end{aligned}
\end{equation}
where the expectation value $\langle\mc{W}\rangle\geq 0$ for all separable states \cite{GUHNE2009detection}.

As a result, here the error threshold is taken as $\epsilon=1-1/d$. From Theorem \ref{Th:optimal}, we know that the optimal verification round is,
\begin{equation}\label{Eq:Ent}
\begin{aligned}
N^{op}= \left\lceil\frac{\ln \delta^{-1}}{\ln\left[1-\nu(\Omega_{op})\epsilon\right]^{-1}}\right\rceil = \left\lceil\frac{\ln \delta^{-1}}{\ln\left(\frac{d+1}{2}\right)}\right\rceil.
\end{aligned}
\end{equation}
Thus EP property can be verified with single round if $d\geq 2\delta^{-1}-1$. Moreover, suppose we consider the verification protocol just with two MUBs, which is the easiest to realize in the experiment, the corresponding verification round is,
\begin{equation}\label{Eq:2MUB}
\begin{aligned}
N^{2-MUB}\leq \left\lceil\frac{\ln\delta^{-1}}{\ln \left(1-\frac{d-1}{2d}\right)^{-1}}\right\rceil=\left\lceil\frac{\ln \delta^{-1}}{\ln\left (\frac{2d}{d+1}\right)}\right\rceil.
\end{aligned}
\end{equation}
Thus we can use two measurement settings, for example, $X$ and $Z$ bases states and measurements to detect the entanglement property of the quantum channel.

\subsection{Quantumness quantification}
In this section, we further apply the verification to the quantification of quantumness. Specifically, an operational measure called the robustness of the quantum memory can be lower bounded using the verification protocol.

We first introduce the robustness of entanglement \cite{vidal1999robustness},
\begin{equation}
\begin{aligned}
\mc{R}^s(\rho):=\min_{\mc{\sigma}\in\mc{S}}\left\{t\geq0,\ \frac{\rho+t\sigma}{1+t}\in\mc{S}\right\},
\end{aligned}
\end{equation}
where $\mc{S}$ is the set of separable states. $\mc{R}(\rho)$ quantifies how much separable noise needs to be introduced to make the state separable. If one allows the noisy state $\sigma$ to be any state, the definition becomes the generalized robustness $\mc{R}_G^s(\rho)$. By definition, $\mc{R}_G^s(\rho)\leq\mc{R}^s(\rho)$.

In a similar way, the robustness of quantum channel is defined as \cite{liu2019resource,yuan2019robustness},
\begin{equation}\label{Eq:robCh}
\begin{aligned}
\mc{R}(\mc{E}):=\min_{\mc{M}\in\mc{F}}\left\{t\geq0,\ \frac{\mc{E}+t\mc{M}}{1+t}\in\mc{F}\right\},
\end{aligned}
\end{equation}
where $\mc{F}$ is the set of EB channels. If one allows the mixed channel $\mc{M}$ to be any channel, the definition becomes the generalized robustness $\mc{R}_G(\mc{E})$. By definition, $\mc{R}_G(\mc{E})\leq\mc{R}(\mc{E})$. The (generalized) robustness measure of quantum channel owns a few of significant operational meaning, such as the amount of classical simulation cost and the advantage in state discrimination-based quantum games.

\begin{observation}\label{Ob:robust}
\begin{equation}\label{}
\begin{aligned}
\mc{R}^s(\Phi_{\mc{E}}) \leq \mc{R}(\mc{E}),\ \mc{R}_G^s(\Phi_{\mc{E}}) \leq \mc{R}_G(\mc{E}),
\end{aligned}
\end{equation}
where $\Phi_{\mc{E}}$ is the Choi state of the quantum channel $\mc{E}$.
\end{observation}
\begin{proof}
Here we prove the first inequality, and the second can be proved in the same way. If we write Eq.~\eqref{Eq:robCh} in the Choi state form, we can see that the noisy Choi state $\Phi_{\mc{M}}$ is not only a separable state but also under the additional constraint---maximally mixed on the first subsystem. However, the minimization of $\mc{R}^s(\Phi_{\mc{E}})$ does not need this constraint thus serves as a lower bound.
\end{proof}
From Observation \ref{Ob:robust}, one has $\mc{R}_G^s(\Phi_{\mc{E}}) \leq \mc{R}_G(\mc{E})\leq \mc{R}(\mc{E})$. Thus we can give a reliable lower bound of the robustness of quantum channel by estimating the corresponding measure on the Choi state. From Ref.~\cite{eisert2007quantitative}, the generalized robustness of entanglement on states can be lower bounded by the witness as,
\begin{equation}
\begin{aligned}
R^s_G(\rho)\geq \frac{|\tr(\mc{W}\rho)|}{\lambda_{max}},
\end{aligned}
\end{equation}
where $\lambda_{max}$ is the largest eigenvalue of the witness operator $\mc{W}$. Inserting the witness in Eq.~\eqref{Eq:witness}, we have
\begin{equation}
\begin{aligned}
R^s_G(\Phi_{\mc{E}})\geq d\tr(\Phi_{\mc{E}}\Phi_+)-1.
\end{aligned}
\end{equation}
As a result, to confirm $\mc{R}(\mc{E})\geq r$, the entanglement infidelity should satisfies $\epsilon\leq \frac{d-r-1}{d}$. And we have the trial number of the optimal strategy given by
\begin{equation}\label{Eq:robNum}
\begin{aligned}
N^{op}= \left\lceil\frac{\ln \delta^{-1}}{\ln\left[1-\nu(\Omega_{op})\epsilon\right]^{-1}}\right\rceil = \left\lceil\frac{\ln \delta^{-1}}{\ln\left(\frac{d+1}{r+2}\right)}\right\rceil.
\end{aligned}
\end{equation}
Note as $r=0$, Eq.~\eqref{Eq:robNum} becomes the result in Eq.~\eqref{Eq:Ent} of the verification of entanglement. We can further reduce the measurement efforts by using less MUBs.

\section{Conclusion and outlook} \label{Sec:conclusion}

In this work, we studied the verification of quantum gates. Based on the Choi representation of quantum channels, we analyze the verification strategies with local state inputs and local measurements without the assistance of extra ancillaries.

In the non-adversarial scenario, the verification performance characterized by the type-II error probability $P(\epsilon,\Omega)$ can be calculated by a semidefinite program. On accunt of the unitary invariance and convexity of the pass probability with respect to $\Omega$, one can prove the optimality of a uniformly mixing strategy $\Omega_{op}$ in Eq.~\eqref{Eq:OpOp}, which can be realized by a CB test with $(d+1)$-MUB when $d$ is a prime power, or other mixing strategy based on quantum state $2$-design. Moreover, we show that the performance of all the Bell-diagonal strategies can be exactly evaluated.

In the adversarial scenario, the verification performance characterized by the entanglement fidelity lower bound $F(N,\delta,\Omega)$ and number of trials upper bound $N(\epsilon,\delta,\Omega)$ are in general hard to solve, while the corresponding state parameters $F_S(N,\delta,\Omega)$ and $N_S(\epsilon,\delta,\Omega)$ can provide a useful bound. We prove that, for the Bell-diagonal strategies with the form in Eq.~\eqref{Eq:adbelldiagonal}, the calculation of $F(N,\delta,\Omega)$ and $N(\epsilon,\delta,\Omega)$ can be reduced to its corresponding state version $F_S(N,\delta,\Omega)$ and $N_S(\epsilon,\delta,\Omega)$. Meanwhile, we prove that, among all the Bell-supported strategies $\Omega$ defined in Section \ref{SSec:adversarial}, for given trial rounds $N$ and significant level $\delta$, the optimal $F(N,\delta,\Omega)$ can always be achieved by the homogeneous strategies.

More specifically, we analyze the local verification strategies and their performance for some common quantum gates, such as single-qubit and qudit gates, multi-qubit Clifford gates, and multi-qubit controlled-$Z$ and controlled-$X$ gates. We also demonstrate the application of gate verification for channels' property testing. We show that gate verification can be used to test the entanglement-preserving property and further the quantification of the robustness of quantum memory.

To enhance the robustness of our work against state preparation and measurement error, we may consider the combination of channel verification with common robust methods, such as randomized benchmarking \cite{emerson2005scalable,magesan2011scalable}, robust tomographic information extraction \cite{kimmel2014robust}, and gate set tomography \cite{blumekohout2013robust}. On the other hand, it is important to make the gate verification protocol robust against few rounds of failure tests \cite{takeuchi2019resource}.

To characterize the quantumness in a channel is currently a hot topic \cite{chitamber2019quantum,Yuan2019hypothesis,gour2018entropy,gour2019comparison,liu2019operational,liu2019resource,yuan2019robustness}. Here we analyze the application of gate verification to quantify the robustness of quantum memory \cite{liu2019resource,yuan2019robustness}, we believe that our method can be extended to quantify other properties of the channel, such as the coherence generating power \cite{Mani2015Cohering,Diaz2018usingreusing}, magic \cite{veitch2014resource} and so on.

During the preparation of this manuscript, we notice two recent related works \cite{liu2019efficient,zhu2019efficientgate}. Comparing to Ref.~\cite{liu2019efficient}, we analytically derive the optimal verification strategy for the general $d$-level unitary. Ref.~\cite{zhu2019efficientgate} develops a very general framework for the quantum gate verification with local state inputs and local measurements, which is suitable for quite a few gates, especially for the multi-partite ones. Here we focus on the preparation and measurement strategies and directly relate them to channel's Choi representation. As a result, our performance (by the  number of trials) on the multi-qubit Clifford gate in Eq.~\eqref{Eq:clif}, $N\leq \left\lceil\frac{2^{2n}-1}{2^{2n-1}}\epsilon^{-1}\ln\delta^{-1}\right\rceil$ is better than the one in Ref.~\cite{zhu2019efficientgate}, $N\leq \left\lceil 3\epsilon^{-1} \ln\delta^{-1}\right\rceil$. Moreover, here we also consider the quantum gate verification in adversarial scenario and its application in channel property testing.

\begin{acknowledgments}
We thanks Xiongfeng Ma and Xiao Yuan for the helpful discussion on the quantumness of channels and the semidefinite programming. This work was supported by the National Natural Science Foundation of China Grants No.~11875173 and No.~11674193, the National Key R\&D Program of China Grants No.~2017YFA0303900 and No.~2017YFA0304004, and the Zhongguancun Haihua Institute for Frontier Information Technology. You Zhou was supported in part by the Templeton Religion Trust under grant TRT 0159.
\end{acknowledgments}


%

\onecolumngrid
\newpage
\begin{appendix}

\section{Heisenberg-Weyl Operators and generalized qudit Bell states} \label{Sec:quditBell}

The Heisenberg-Weyl group is a generalization of Pauli group. For a qudit Hilbert space with computational basis $\{\ket{l}\}_{l=0}^{d-1}$, we define
\begin{equation}
\begin{aligned}
Z &= \sum_{l=0}^{d-1} \exp(i\frac{2\pi}{d}l) \ket{l}\bra{l}, \\
X &= \sum_{l=0}^{d-1} \ket{l+1}\bra{l}, \\
\end{aligned}
\end{equation}
here $\ket{l+1}$ means $\ket{(l+1) \text{ mod } d}$.

The Heisenberg-Weyl operator $W(u,v)$ is defined to be
\begin{equation}
W(u,v) = X^u Z^v,
\end{equation}
with $u,v=0,1,...,d-1$. It is easy to verify that
\begin{equation}\label{Eq:Acomm}
\begin{aligned}
X^d=&Z^d=I, \quad (X^u)^\dagger = X^{-u}, \quad (Z^v)^\dagger = Z^{-v}, \\
X^uZ^v &= \exp\left(-i\frac{2\pi}{d}uv\right)Z^vX^u, \\
W(u,v)W(u',v') &= \exp\left(-i\frac{2\pi}{d}(uv' - vu')\right) W(u',v')W(u,v). \\
\end{aligned}
\end{equation}

Define $\Phi_{0,0} = \Phi_+ = \dfrac{1}{\sqrt{d}}\sum_{j=0}^{d-1} \ket{jj}$. The generalized qudit Bell states\cite{bennett1993Teleporting} are
\begin{equation}
\begin{aligned}
\ket{\Phi_{u,v}} &:= W(u,v)\ket{\Phi_+} \\
&= \frac{1}{\sqrt{d}}\sum_{l=0}^{d-1} \exp\left(\frac{2\pi i}{d}lv\right)\ket{l}_A\otimes \ket{l+u}_B,
\end{aligned}
\end{equation}
Denote $\Phi_{u,v}:=\ket{\Phi_{u,v}}\bra{\Phi_{u,v}}$. The qudit Bell states $\{\Phi_{u,v}\}_{u,v=0}^{d-1}$ form an orthonomal basis,
\begin{equation}
\begin{aligned}
\braket{\Phi_{u,v}|\Phi_{u',v'}} &= \bra{\Phi_+}(I\otimes W(u,v)^\dag W(u',v') )\ket{\Phi_+} \\
&= \exp\left(-i\frac{2\pi}{d}u_d v\right) \bra{\Phi_+}(I \otimes X^{u_d}Z^{v_d})\ket{\Phi_+} \\
&= \frac{1}{d} \exp\left(-i\frac{2\pi}{d}u_d v\right) \sum_{j=0}^{d-1} \sum_{m,l=0}^{d-1} \exp\left(i\frac{2\pi}{d}v_d l \right) \braket{j,j|m,l+u_d} \braket{m,l|j,j} \\
&= \frac{1}{d} \delta_{u_d,0} \exp\left(-i\frac{2\pi}{d}u_d v\right) \sum_{j=0}^{d-1} \exp\left(i\frac{2\pi}{d}v_d j \right)  \\
&= \delta_{u_d,0} \delta_{v_d,0}
\end{aligned}
\end{equation}
where $u_d:= u'-u, v_d:=v'-v$.

\section{Proof of Lemma \ref{Lm:Belldiag}} \label{Sec:proofs2}

\begin{proof}
The summation in Eq.~\eqref{Eq:bellT} is a``twirling" operation on the Weyl operators, and we frist prove that the twirling result is in the Bell-diagonal form. To prove this, we take out an operator element $\ket{\Phi_{w_1}}\bra{\Phi_{w_2}}$ in the Bell basis with $w_i=(u_i,v_i)$ and $\ket{\Phi_{w_i}}=\mathbb{I}\otimes W_i\ket{\Phi_+}$, and show that it vanishes after the twirling unless $w_1=w_2$. For simplicity, we denote $a=\exp\left(-i\frac{2\pi}{d}\right)$ and $(w,w')=(uv' - vu')$.
\begin{equation} \label{}
\begin{aligned}
&\sum_w \mc{W}^*(u,v) \otimes \mc{W}(u,v) (\ket{\Phi_{w_1}}\bra{\Phi_{w_2}})\\
= &\sum_w W^* \otimes W (\ket{\Phi_{w_1}}\bra{\Phi_{w_2}})W^T \otimes W^{\dag}\\
= &\sum_w W^* \otimes W \mathbb{I}\otimes W_1\ket{\Phi_+}\bra{\Phi_+}\mathbb{I}\otimes W_2^{\dag}W^T \otimes W^{\dag}\\
=&\sum_w a^{(w,w_1)} a^{-(w, w_2)}\mathbb{I}\otimes W_1 \sum_w W^* \otimes W   \ket{\Phi_+}\bra{\Phi_+}W^T \otimes W^{\dag}\mathbb{I}\otimes W_2^{\dag}\\
=&\sum_w e^{(w,w_1-w_2)} \mathbb{I}\otimes W_1 \ket{\Phi_+}\bra{\Phi_+}\mathbb{I}\otimes W_2^{\dag}\\
=&\sum_w e^{(w,w_1-w_2)}\ket{\Phi_{w_1}}\bra{\Phi_{w_2}}\\
=&\sum_{\{u,v\}} e^{-i\frac{2\pi}{d}(u\delta v'-v\delta u')} \ket{\Phi_{w_1}}\bra{\Phi_{w_2}}\\
=&\delta_{\{w_1-w_2\}} \ket{\Phi_{w_1}}\bra{\Phi_{w_2}}.
\end{aligned}
\end{equation}
where $w_1-w_2=(\delta u', \delta v')$. Here the first equality we apply the commuting relation in Eq.~\eqref{Eq:Acomm}; the fourth equality is due to the invariance of the maximally entangled state $\ket{\Phi_+}$ under the operation $W^* \otimes W$.

Then we prove the non-increasing of the passing probability $P(\epsilon,\Omega)$. Note that the twirling operation is a mixing of $d^2$ verification operators $\Omega_{\{u,v\}}=\mc{W}^*(u,v)\otimes \mc{W}(u,v) (\Omega)$ with equal probability $\Omega'=1/d^2\sum \Omega_{\{u,v\}}$. Thus, combing Observation \ref{Ob:convex} and Lemma \ref{Lm:Uinvary}, one has $P(\epsilon,\Omega')\leq 1/d^2\sum P(\epsilon,\Omega')=P(\epsilon,\Omega)$.
\end{proof}

\section{Proof of Lemma \ref{Lm:Bellsupport} } \label{Sec:proofs3}

\begin{proof}
For the strategy $\Omega$, we take a group of eigenvectors $\{\Psi_{j,l}\}$ corresponding to different eigenvalues $\{\lambda_j\}$. If the rank of $\Pi_{j}$ is larger than 1, then $l$ denotes the index in the degenerated space. We set $\Psi_{j,0}$ to be (one of) the Bell state in $\Pi_j$ if $\Pi_j\in \mb{S}_0(\Omega)$. Obviously, $\{\Psi_{j,l}\}$ are also the eigenvectors of $\Omega'$. We denote the set of maximally entangled basis in it as $\Theta(\Psi_{j,l})$.

Similar to the argument in the proof of Observation \ref{Ob:adBelldiag}, we now introduce the permutation-invariant basis
\begin{equation} \label{Eq:Psik}
\Psi_{\mb{k}} = \hat{ \mb{P} }_S\left( \bigotimes_{j,l} \Psi^{\otimes k_{j,l}}_{j,l} \right),
\end{equation}
where $\hat{\mb{P}}_S$ is the symmetrization operator, mixing all possible permutation with respect to different rounds, $k:=[k_{0,0}, k_{0,1}, ..., k_{J-1,L-1}]$ is a sequence of nonnegative integer number with $\sum_{j,l} k_{j,l} = N+1$. If $k$ is non-zero only on the set $\Theta(\Psi_{j,l})$, the generated symmetric state $\Psi_{\mb{k}}$ will also be the maximally entangled state. We denote the set of such $\Psi_{\mb{k}}$ as the symmetric Bell basis $\Theta_\mb{S}(\Psi_{j,l},N)$.

Since $p_{\mc{E}}$ and $f_{\mc{E}}$ in Eqs.~\eqref{Eq:pmcE},~\eqref{Eq:fmcE} only depend on the diagonal elements of $\Phi_{\mc{E}}$ in the basis of $\Omega$, without loss of generality, we may assume that the Choi state is diagonal in the product basis of $\{\Psi_{j,l}\}$. We only need to consider the Choi state $\Phi_{\mc{E}}$ as the mixture of $\Psi_{\mb{k}}$
\begin{equation} \label{Eq:PhiEc}
\Phi_{\mc{E}}(\mb{c}) = \sum_{\mb{k}\in\mb{K}} c_{\mb{k}} \Psi_{\mb{k}},
\end{equation}
where $\mb{c}=\{c_{\mb{k}}\}$ are the nonnegative mixing coefficients with $\sum_{k\in\mb{K}} c_{\mb{k}} = 1$, and $\mb{K}$ is the set of all possible $\mb{k}$. Since $\Psi_{\mb{k}}$ might not meet the requirement of Choi state, there is extra limitation on the coefficients:
\begin{equation} \label{Eq:coefflimitation}
\tr_B[\Phi_{\mc{E}}(\mb{c})] = \left(\frac{\mathbb{I}_d}{d}\right)^{\otimes (N+1)}.
\end{equation}
We denote the set of legal coefficients $\mb{c}$ satisfying Eq.~\eqref{Eq:coefflimitation} as $\mb{C}(\Psi_\mb{k})$, which is determined by $\{\Psi_\mb{k}\}$. Note that, due to the linearity of Eq.~\eqref{Eq:coefflimitation}, $\mb{C}(\Psi_\mb{k})$ is a convex set.

According to Eq.~\eqref{Eq:pmcE},~\eqref{Eq:fmcE},~\eqref{Eq:Psik}, and \eqref{Eq:PhiEc}, one have
\begin{equation}
\begin{aligned}
p_\mc{E}(\mb{c}) &= \sum_{k\in\mb{K}} c_{\mb{k}} \eta_{\mb{k}}(\vec{\lambda}), \quad \mb{c}\in\mb{C}(\Psi_\mb{k})\\
f_\mc{E}(\mb{c}) &= \sum_{k\in\mb{K}} c_{\mb{k}} \zeta_{\mb{k}}(\vec{\lambda}), \quad \mb{c}\in\mb{C}(\Psi_\mb{k})
\end{aligned}
\end{equation}
where $\vec{\lambda}:= (\lambda_{0,0}, \lambda_{0,1}, ..., \lambda_{J-1,L-1})$ is the eigenvalues of $\Omega$ or $\Omega'$, and
\begin{equation}
\begin{aligned}
\eta_{\mb{k}}(\vec{\lambda}) &:= p(\mb{k}) = \sum_{i|k_i>0}\frac{k_i}{(N+1)}\lambda_i^{k_i-1}\Pi_{j\neq i|k_j>0} \lambda_j^{k_j}, \\
\zeta_{\mb{k}}(\vec{\lambda}) &:= f(\mb{k}) = \frac{k_1}{(N+1)}\Pi_{i|k_i>0} \lambda_i^{k_i}.
\end{aligned}
\end{equation}
Here $\lambda_i^0$ is set to be $1$, even if $\lambda_i=0$. Due to the degeneration of $\{\lambda_{j,l}\}$, for different $\mb{k}$, the values of $\eta_{\mb{k}}(\vec{\lambda})$ and $\zeta_{\mb{k}}(\vec{\lambda})$ could be the same. The optimization value of $F(N,\delta,\Omega)$ is determined by the 2-D region of $(p_\mc{E}(\mb{c}), f_\mc{E}(\mb{c}))$ with legal $\mb{c}\in\mb{C}(\Psi_\mb{k})$.

Our main idea to prove $F(N,\delta,\Omega') \geq F(N,\delta,\Omega)$ is to show that the optimizing area of $\Omega'$ belongs to the optimizing area of $\Omega$, that is, the point $(p_\mc{E}(\mb{c}), f_\mc{E}(\mb{c}))$ by coefficients $\mb{c}$ with $\Omega'$ can always be achieved by the same coefficients $\mb{c}$ with $\Omega$.

First, for the strategy $\Omega'$, all the value of $(p_\mc{E}(\mb{c}), f_\mc{E}(\mb{c}))$ with $\mb{c}\in\mb{C}(\Psi_\mb{k})$ can be achieved even if we only consider the symmetric Bell basis $\Psi_\mb{k}\in\Theta_\mb{S}(\Psi_{j,l},N)$. Since the symmetric bell basis terms $\{\Psi_\mb{k}\}$ naturally satisfy Eq.~\eqref{Eq:coefflimitation}, the Bell-coefficients $\{c_\mb{k}\}$ can then be chosen freely, without any extra requirements than non-negative and normalization. On the other hand, due to the degeneracy of eigenvalues, i.e., $\tilde{\lambda}_j\in\lambda(\mb{S}_0(\Omega))$, all the values of $\eta_{\mb{k}}(\vec{\lambda})$ and $\zeta_{\mb{k}}(\vec{\lambda})$ can be realized by the symmetric Bell basis set $\Theta_\mb{S}(\Psi_{j,l},N)$.

Second, for each symmetric Bell strategy $\{c_\mb{k}\}$ of $\Omega'$, one can realize it on $\Omega$ with the same value of $(p_\mc{E}(\mb{c}), f_\mc{E}(\mb{c}))$. Note that the feasible coefficients region $\mb{C}(\Psi_\mb{k})$ for $\Omega$ and $\Omega'$ are the same. Moreover, all the values of $\eta_{\mb{k}}(\vec{\lambda})$ and $\zeta_{\mb{k}}(\vec{\lambda})$ for the symmetric Bell basis are the same.
\end{proof}

\end{appendix}

\end{document}